\documentclass[graybox]{svmult} 

\usepackage{type1cm}         

\usepackage{makeidx}         
\usepackage{graphicx}        
\usepackage{multicol}        
\usepackage[bottom]{footmisc}

\usepackage{newtxtext}        
\usepackage{newtxmath}       

\usepackage{hyperref}
\usepackage{amsmath}
\usepackage{amscd}
\usepackage{wasysym}

\DeclareMathOperator{\Ker}{Ker}
\DeclareMathOperator{\Id}{Id}
\newcommand{\V}{\mathbb V}
\newcommand{\W}{\mathbb W}
\newcommand{\symb}{\it{symb}}
\newcommand{\Hom}{\mathrm{Hom}}
\newcommand{\End}{\mathrm{End}}

\newcommand{\SO}{\mathrm{SO}}
\newcommand{\Spin}{\mathrm{Spin}}
\newcommand{\CSpin}{\mathrm{CSpin}}
\newcommand{\Cl}{C\ell}

\newcommand{\stab}{\mathfrak{stab}}

\newcommand{\fso}{\mathfrak{so}}
\newcommand{\fsu}{\mathfrak{su}}

\newcommand{\fg}{\mathfrak{g}}
\newcommand{\fh}{\mathfrak{h}}

\newcommand{\fk}{\mathfrak{k}}
\newcommand{\fp}{\mathfrak{p}}

\newcommand{\fa}{\mathfrak{a}}
\newcommand{\fD}{\mathfrak{D}}

\newcommand{\ZZ}{\mathbb{Z}}
\newcommand{\CC}{\mathbb{C}}
\newcommand{\be}{\boldsymbol{e}}

\newcommand{\vol}{\operatorname{vol}}

\makeindex             

\begin{document}

\title*{Remarks on highly supersymmetric backgrounds of $11$-dimensional supergravity}
\author{Andrea Santi}
\institute{Andrea Santi \at University of Padova, Mathematics Department ``Tullio Levi-Civita'', 35121 Padova (Italy), \email{asanti.math@gmail.com}
}
%
\maketitle
\abstract{This note focuses on some properties and uses of filtered deformations in the context of $D=11$ supergravity. We define the concept of abstract symbol and
give a strong version of the Reconstruction Theorem, namely a bijective correspondence 
from the space of highly supersymmetric supergravity backgrounds to the space of abstract symbols.
We propose a general strategy to construct highly supersymmetric supergravity backgrounds
and present an example in detail, which includes the computation of the ideal generated by the Killing spinors of two known
pp-wave backgrounds with $N=24$ supersymmetry. 
Finally, we give an alternative proof, based on the isotropy algebra of a supergravity background, of 
a classical supersymmetry gap result of Gran, Gutowski, Papadopoulos and Roest. 
}
\section{Introduction}
\label{sec:introduction}
An important line of research in string and M-theory is the construction of backgrounds of their low-energy effective counterparts, i.e., supergravity, since this is where 
strings and higher-dimensional branes can propagate.
Backgrounds preserving maximal or near to maximal supersymmetry are important classes of such solutions.

In particular, the classical limit of M-theory is supergravity in $D=11$ dimensions. 
From a geometric perspective, a background of $D=11$ supergravity consists 
of an $11$-dimensional Lorentzian spin manifold $(M,g)$ with a closed $4$-form $F\in\Omega^4(M)$, satisfying the following
coupled system of PDEs:
\begin{equation}
    \label{eq:bosfieldeqs}
    \begin{split}
      d\star F &= \tfrac12 F\wedge F\;,\\
      \operatorname{Ric}(X,Y) &=\tfrac12 g(\imath_XF,\imath_Y F) - \tfrac16
      \|F\|^2 g(X,Y)\;,
    \end{split}
  \end{equation}
  for all $X,Y\in\mathfrak{X}(M)$. These are 
	the bosonic field
  equations of $D=11$ supergravity. 
	
	Let $S(M)\to M$ be the spinor bundle of $(M,g)$, with typical fiber the irreducible module $S\cong\mathbb R^{32}$ of the Clifford algebra $\Cl(V)\cong 2\mathbb R(32)$.
	A spinor field $\varepsilon\in\Gamma(S(M))$ is called a Killing spinor if we have
	\begin{equation}
	\label{eq:Killingspinor-equation}
	\nabla_X \varepsilon-\tfrac1{24} (X \cdot F - 3 F \cdot X)\cdot\varepsilon=0
	\end{equation}
for all $X\in\mathfrak X(M)$, where $\nabla$ is the Levi-Civita connection and $\cdot$ Clifford multiplication.
The connection on $S(M)$ defined by the L.H.S. of 	\eqref{eq:Killingspinor-equation} is called the ``superconnection'' and the amount of spinors it preserves, i.e., Killing spinors, is an important invariant of a supergravity background. 

Backgrounds with $16<N\leq  32$ Killing spinors are called {\it highly supersymmetric} and their classification 
is widely open. Although this might be a distant goal, a number of different
solutions have been proposed in the literature:
\begin{itemize}
	\item Maximally supersymmetric backgrounds, i.e., with $N = 32$ Killing spinors,
	\item pp-waves backgrounds, i.e., Brinkmann spaces with flat transverse geometry,
	\item G\"odel backgrounds, i.e., solutions admitting closed timelike curves.
\end{itemize}
Maximally supersymmetric backgrounds were classified in \cite{FOFP} up to local isometry (see \cite{FOFS2017} for a purely Lie theoretic proof in the spirit of this note): one has Minkowski spacetime, the Freund-Rubin backgrounds $Ad_4\times S^{7}$ and $Ad_7\times S^{4}$, and the 
Kowalski-Glikman symmetric pp-wave.
A supersymmetry gap result is also known:
\begin{theorem}
\cite{GGPR}
\label{thm:No-Go31}
If a background of $D=11$ supergravity has at least $31$ Killing spinors, then it is locally isometric to a maximally supersymmetric background.
\end{theorem}
The proof relies on a careful analysis of the integrability conditions on the curvature of the superconnection that arise 
from Killing spinor equations, bosonic field equations, Bianchi identities of the Riemann curvature and $dF=0$. The result was further  extended to $N=30$ in \cite{GGP}. At present, the highest number of Killing spinors
known for a non-maximally supersymmetric background is $N=26$, for 
the pp-wave of \cite{M}.
Highly supersymmetric pp-waves are known with any even number $18\leq N\leq 24$ of Killing spinors \cite{GH}, whereas typically $N=20$  
for G\"odel backgrounds \cite{HT}.

Few general structural results were known until recently:
\begin{theorem} 
\label{thm:local-homogeneity}
\cite{FOFH, FOFMP} Let $(M,g,F)$ be a background of $D=11$ supergravity. Then:
\begin{enumerate}
\item There exists an associated Lie superalgebra $\fk=\fk_{\bar 0}\oplus \fk_{\bar 1}$,
where $\fk_{\bar 0}$ is the space of Killing vectors preserving $F$ and $\fk_{\bar 1}$ the space of Killing spinors;
\item If $(M,g,F)$ is highly supersymmetric then it is locally
homogeneous.
\end{enumerate}
\end{theorem}
The Lie superalgebra $\fk=\fk_{\bar 0}\oplus \fk_{\bar 1}$ is called the Killing superalgebra of the background. Its odd part generates an ideal $[\fk_{\bar 1},\fk_{\bar 1}]\oplus\fk_{\bar 1}$, which in this paper call the {\it transvection superalgebra}, following a similar terminology used in the context of symmetric spaces. Clearly $N=\dim\fk_{\bar 1}$ and  $\dim\fk_{\bar 1}>16$ precisely in the highly supersymmetric regime. 
These results paved the way to use Lie algebraic techniques, but unfortunately many
homogeneous supergravity backgrounds are not highly supersymmetric. 

Inspired by the superspace presentation of supergravity theories in terms of Tanaka structures \cite{SSI, SSII}, the precise algebraic structure of Killing superalgebras has been derived and further structural results obtained:
\begin{theorem}\cite{FOFS2017, FOFS2017II}
\label{thm:JA}
  Let $(M,g,F)$ be an $11$-dimensional Lorentzian spin manifold 
  with a closed $4$-form $F\in\Omega^4(M)$.  If $\dim\fk_{\bar 1} > 16$, then:
	\begin{enumerate}			
	\item The bosonic field
  equations    \eqref{eq:bosfieldeqs} of $D=11$ supergravity are automatically satisfied;	
	\item The associated Killing superalgebra is (isomorphic to) a filtered subdeformation of the Poincar\'e superalgebra, the same holds for the transvection superalgebra;
\item The background is fully determined, up to local isometry, by the pair $(\varphi, S')$, where 
\begin{equation}
\label{eq:pair-introduction}
\varphi=F^\sharp|_o\in\Lambda^4 V\quad\text{and}\quad S'=\fk_{\bar 1}|_o\subset S\;.
\end{equation}
In particular if $\varphi=0$ then $(M,g)$ is locally isometric to Minkowski spacetime. 
		\end{enumerate}
\end{theorem}
We call the pair $\symb(M,g,F)=(\varphi,S')$ given in \eqref{eq:pair-introduction} the	{\it geometric symbol} of $(M,g,F)$.

This note aims to highlight some aspects and uses of the correspondence between highly supersymmetric supergravity backgrounds and transvection superalgebras. 

We will present an alternative and relatively short proof of Theorem \ref{thm:No-Go31}, which does not use the restrictions on the holonomy of the superconnection but it is based on the geometric symbol. 
More precisely, we do not require the Bianchi identities of the Riemann curvature or $dF=0$ (equivalently, the bosonic field equations)
but resort to one of the first properties satisfied by the isotropy of the transvection superalgebra of the background, and a partial knowledge of the space of bispinors.

The usage of filtered subdeformations to construct backgrounds has so far been hindered by the amount of algebraic data required to reconstruct the deformation from the symbol. We significantly simplify this data in Theorem \ref{thm:main-reduction}, paving the way
to the generation of new highly supersymmetric backgrounds with prescribed symmetries. The proposed method takes 
$\SO(V)$-orbits of four-vectors $\varphi\in\Lambda^4 V$ as starting point and their rank as a useful organizing principle. 
To date, up to my knowledge, there is no known highly supersymmetric background for which the rank is maximal, i.e., $\operatorname{rk}(\varphi)=11$, and all known backgrounds with $\operatorname{rk}(\varphi)=10$ have $N\leq 22$.
On the other hand, those with $N\geq 24$ have rank relatively small: $\operatorname{rk}(\varphi)\leq 8$. 
Roughly speaking, we might say that the more Killing spinors there are, the more the underlying geometry ought to be rigid and the four-vector less intricate.

The paper has been organized as follows.
We first review in \S\ref{sec:2} and \S\ref{sec:3} the basics of filtered subdeformations of the Poincar\'e superalgebra, including the Reconstruction Theorem \cite{FOFS2017II}
and the concepts of Dirac kernel and Lie pair also introduced  in \cite{FOFS2017II}.
A Lie pair is a pair $(\varphi,S')$ satisfying a system of coupled algebraic equations, quadratic on $\varphi$ and cubic on $S'$, which allows to recover the isotropy algebra of the transvection superalgebra.
We then prove Theorem \ref{thm:No-Go31} using Dirac kernels and Lie pairs in \S\ref{sec:4}. 

In \S\ref{sec:class-probl-highly}, we turn to study   the reconstruction problem for supergravity backgrounds: we prove Theorem
\ref{thm:main-reduction}, introduce the concept of abstract symbol in Definition \ref{def:main} and state the strong version of the Reconstruction Theorem in Theorem \ref{thm:strong-version-reconstruction}.

The last sections \S\ref{sec:6} and \S\ref{sec:7} give details of the strategy to construct supergravity backgrounds in an example, in particular we describe the transvection superalgebras of two symmetric pp-waves backgrounds with $N=24$ Killing spinors discovered in \cite{GH}. 
\section{Filtered deformations and the reconstruction theorem}
\label{sec:2}
Let $(V,\eta)$ be the Lorentzian vector space with ``mostly minus'' signature $(1,10)$
and $S$ the spinor representation of $\fso(V)$.
We recall that $S$ has an
$\fso(V)$-invariant symplectic structure $\left<-,-\right>$ such that
$\left<v\cdot s_1, s_2\right> = - \left<s_1, v \cdot s_2\right>$
for all $s_1,s_2 \in S$ and $v \in V$.

The {\it Poincar\'e superalgebra} $\fp$ has underlying vector space
$\fso(V) \oplus S \oplus V$ and nonzero Lie brackets
\begin{equation*}
    [A,B] = AB - BA\;, \qquad [A,s] = \sigma(A)s\;, \qquad [A,v] = Av\;, \qquad
    [s,s] = \kappa(s,s)\;,
\end{equation*}
where $A,B \in \fso(V)$, $v \in V$ and $s \in S$. Here $\sigma$ is the spinor representation of $\fso(V)$ and
$\kappa: \odot^2 S \to V$ is the so-called \emph{Dirac current},
defined by
\begin{equation*} 
  \eta(\kappa(s,s),v) = \left<s, v\cdot s\right>,
\end{equation*}
for all $s\in S$ and $v\in V$.  One important property of
$\kappa$ is that its restriction to
$\odot^2 S'$ is surjective on $V$, for any subspace $S'\subset S$ with $\dim S' > 16$ \cite{FOFH}. 
This algebraic fact is usually referred to as the ``local Homogeneity Theorem'', due to the
role it plays in the proof of Theorem \ref{thm:local-homogeneity}. If we grade $\fp$ so that $\fso(V)$, $S$ and $V$ have degrees
$0$, $-1$ and $-2$, respectively, then the above Lie brackets turn
$\fp$ into a $\ZZ$-graded Lie superalgebra
\begin{equation*}
  \fp = \fp_{0} \oplus \fp_{-1} \oplus \fp_{-2}, \qquad \fp_0=\fso(V),
  \qquad \fp_{-1}=S,\qquad\fp_{-2}=V.
\end{equation*}
Note that the parity is consistent with the $\mathbb Z$-grading, since
$\fp_{\bar 0} = \fp_0 \oplus \fp_{-2}$ and $\fp_{\bar 1} = \fp_{-1}$.

Let $\fa=\fh\oplus S'\oplus V$ be a graded subalgebra of $\fp$ with $\dim S'>16$ and $\fg$ a {\it filtered deformation} of $\fa$, i.e., the Lie brackets of $\fg$ take the following general form \cite{CK}:
\begin{equation}
\label{eq:generalbrackets}
  \begin{aligned}[m]
    [A,B]&=AB-BA\\
    [A,s]&=\sigma(A)s\\
    [A,v]&=Av+\delta(A,v)
  \end{aligned}
  \qquad\qquad
  \begin{aligned}[m]
    [s,s]&=\kappa(s,s)+\gamma(s,s)\\
    [v,s]&=\beta(v,s)\\
    [v,w]&=\alpha(v,w)+\rho(v,w)
  \end{aligned}
\end{equation}
where $A,B\in \fh$, $v,w\in V$ and $s\in S'$. Here $\alpha\in\Hom(\Lambda^2 V,V)$,
$\beta\in\Hom(V\otimes S',S')$, $\gamma\in\Hom(\odot^2 S',\fh)$ and
$\delta\in\Hom(\fh\otimes V,\fh)$ are maps of degree $2$, and
$\rho\in\Hom(\wedge^2V,\fh)$ of degree $4$. If we do not  mention the subalgebra $\fa$ explicitly, we simply say that
$\fg$ is a \emph{highly supersymmetric filtered subdeformation} of $\fp$. Finally, we say that 
$\fg=\fg_{\bar 0}\oplus\fg_{\bar 1}$ is {\it odd-generated} if $\fg_{\bar 0}=[\fg_{\bar 1},\fg_{\bar 1}]$. 

The notion of isomorphism $\Phi:\fg\to\widetilde\fg$ of highly supersymmetric filtered subdeformations is given in \cite[Def. 5]{FOFS2017II}: 
it suffices to know that
\begin{equation}
  \label{eq:generaliso}
  \Phi(A)=g\cdot A,\qquad
  \Phi(s)=g\cdot s,\qquad\text{and}\qquad
  \Phi(v)=g\cdot v+ X_v,
\end{equation}
for some $g\in\Spin(V)$, $X:V\to\widetilde\fh$. The notion of embedding is given in \cite[Def. 11]{FOFS2017II}.

As we now explain, we are interested in filtered deformations that are {\it realizable}, that is, corresponding to some highly supersymmetric supergravity background.

Associated to any $\varphi\in\Lambda^4 V$, there are two natural maps $\beta^\varphi:V\otimes S\to S$ and
$\gamma^\varphi:\odot^2 S\to\fso(V)$, defined by
\begin{align}
\label{eq:thetatoo}
\beta^{\varphi}(v,s)&=\tfrac1{24}(v\cdot\varphi-3\varphi\cdot v)\cdot s\;,\\
\label{eq:thetatoo-II}
\gamma^{\varphi}(s,s)v&=-2\kappa(\beta^\varphi(v,s),s)\;,
\end{align}
for all $s\in S$, $v\in V$. We will sometimes use also the notation $\beta^\varphi_v(s) = \beta^\varphi(v,s)$. 
As proven in \cite[Proposition ~7]{FOFS2017}, these maps are characterized by the ``cocyle condition''
\begin{equation}
\label{eq:secondcocyclecondition}
\sigma(\gamma^\varphi(s,s)) s = - \beta^\varphi(\kappa(s,s),s)\;,
\end{equation}
for all $s\in S$, and it turns out that the Spencer cohomology group $H^{2,2}(\fp_-,\fp)\cong\Lambda^4 V$ precisely via \eqref{eq:thetatoo}-\eqref{eq:thetatoo-II}. (In general this cohomology group encodes the Killing spinor equations for supersymmetric field theories, see \cite{FOFS2017, MFSI, MFSII} for details.)

Now, it is well-known that 
$\odot^2 S\cong\Lambda^1 V\oplus \Lambda^2 V\oplus\Lambda^5 V$
as an $\fso(V)$-module. This decomposition is unique, since all 
the three summands are $\fso(V)$-irreducible and inequivalent, so we will
consider $\Lambda^q V$ directly as a
subspace of $\odot^2 S$, for $q=1,2,5$. We decompose any element $\omega\in\odot^2 S$ into
$\omega=-\tfrac{1}{32}\big(\omega^{(1)}+\omega^{(2)}+\omega^{(5)}\big)$ accordingly,
where $\omega^{(q)}\in\Lambda^q V$ for $q=1,2,5$. The overall factor of $-\tfrac{1}{32}$
is introduced so that $\omega^{(1)}$
coincides exactly
with the Dirac current of $\omega$.
We may then re-write \eqref{eq:thetatoo-II} as:
\begin{equation}
\label{definition-gamma-useful}
\begin{aligned}
\eta(\gamma^\varphi(\omega)v,w)&=\tfrac{1}{6}\langle s, \big(2\imath_v\imath_w\varphi-v\wedge w\wedge \varphi\big)\cdot s\rangle   \\
&=\tfrac{1}{3}\eta(\imath_v\imath_w\varphi,\omega^{(2)})
+\tfrac{1}{6}\eta(\imath_v\imath_w\star\varphi,\omega^{(5)})\;,
\end{aligned}
\end{equation}
for all $v,w\in V$, where 
musical isomorphisms have been tacitly used. (This will also be the case throughout the whole paper, without further mention.)
Evidently one has $\Ker\gamma^\varphi\supset\Lambda^1 V$.
\begin{definition}\cite{FOFS2017II}
\label{def:realizable}
A highly supersymmetric filtered subdeformation $\fg$ of $\fp$ is {\it realizable} if there exists $\varphi\in\Lambda^4 V$ such that:
\begin{enumerate}
	\item $\varphi$ is $\fh$-invariant;
	\item $\varphi$ is closed, i.e., 
	\begin{equation}
	\label{eq:closed}
	d\varphi(v_0,\ldots,v_4) = \sum_{i<j}(-1)^{i+j}
      \varphi(\alpha(v_i,v_j), v_0, \ldots, \hat v_i,\ldots,\hat
      v_j,\ldots,v_4) = 0
    \end{equation}
    for all $v_0,\ldots, v_4\in V$;
	\item The components of the Lie brackets of $\fg$ of degree $2$ are of the form
	\begin{equation}
  \label{eq:KSAina}
  \begin{aligned}[m]
	 \alpha(v,w) &= X_v w - X_w v\\
   \gamma(s,s) &= \gamma^\varphi(s,s) - X_{\kappa(s,s)}
  \end{aligned}
  \qquad\qquad
  \begin{aligned}[m]  
	\beta(v,s) &= \beta^\varphi(v,s) + \sigma(X_v)s\\
  \delta(A,v) &= [A,X_v] - X_{Av}  
\end{aligned}
\end{equation}
for some linear map $X:V\to\fso(V)$, where $A,B \in \fh$, $v,w \in V$ and $s \in
S'$.
\end{enumerate}
\end{definition}
It is a non-trivial fact that a realizable $\fg$ has a {\it unique} $\varphi\in\Lambda^4 V$ 
that satisfies (1)-(3) of Definition \ref{def:realizable} and that the component $\rho$ of degree $4$ is fully determined by those of degree $2$.
(The idea is that $\varphi$ is an element of the cohomology group $H^{2,2}(\fa_-,\fa)$ and, since $H^{4,2}(\fa_-,\fa)=0$, this group fully determines the subdeformation.) 
Changing the map $X:V\to\fso(V)$ by some values in $\fh$ gives isomorphic filtered subdeformations.

For simplicity of exposition, we will call {\it tamed} any filtered subdeformation of $\fp$ that is highly supersymmetric, odd-generated and realizable. 
We also let 
\begin{equation*}
\label{eq:1:1correspondence}
\begin{aligned}
\mathcal{SB}&=\frac{\big\{\text{highly supersymmetric supergravity bkgds}\; (M,g,F)\big\}}{\text{local isometry}}\\
\mathcal{FD}&=\frac{\big\{\text{maximal tamed filtered subdeformations $\mathfrak g$ of $\mathfrak p$}\big\}}{\text{isomorphism}}
\end{aligned}
\end{equation*}
be the moduli spaces of highly supersymmetric backgrounds and
maximal tamed filtered subdeformations. The following result was proved under the assumption that the isotropy group is closed, but it can be easily relaxed (see, e.g., \cite{Patra}).
\begin{theorem}[Reconstruction Theorem.]\cite{FOFS2017II}
\label{thm:reconstruction}
\begin{enumerate}
	\item The assignment	$
	\mathcal SB\longrightarrow\mathcal FD
	$ that sends a highly supersymmetric supergravity background to its transvection superalgebra is a $1:1$ correspondence;
\item The curvature $R:\Lambda^2 V\to\fso(V)$ of the supergravity background associated to a $\fg\in\mathcal{FD}$ is given by
$R(v,w)=\rho(v,w)-[X_v,X_w]+X_{\alpha(v,w)}$, for all $v,w\in V$.
\end{enumerate}
\end{theorem}
\section{Dirac kernels and Lie pairs}
\label{sec:3}
Let $\fg$ be a tamed filtered subdeformation of $\fp$, with associated graded Lie algebra $\fa=\fh\oplus S'\oplus V$ and $\varphi\in\Lambda^4 V$. Due to Theorem \ref{thm:reconstruction}, we may call $(\varphi,S')$ the symbol of $\fg$.
Since $S'$ has dimension $\dim S'>16$, then
$\odot^2 S'\subset \odot^2 S$ projects surjectively on $\Lambda^1 V$
via the Dirac current operation.  The embedding
$
  \odot^2 S'\subset \odot^2 S\cong\Lambda^1V\oplus\Lambda^2
  V\oplus\Lambda^5 V
$
is in general diagonal and one cannot expect  $\odot^2 S'$ to contain
$\Lambda^q V$, not even if $q=1$. Restricting the
Dirac current to $\odot^2 S'$ gives rise to
a short exact sequence
\begin{equation*}
  \begin{CD} 0 @>>> \fD @>>> \odot^2 S' @>\kappa>> V @>>> 0\;,
  \end{CD}
\end{equation*}
where 
\begin{equation}
\label{eq:Dirac-kernel-definition}
  \begin{split} \mathfrak{D} &= \odot^2 S'\cap (\Lambda^2
    V\oplus\Lambda^5 V)=\left\{\omega\in\odot^2
      S'~\middle|~\omega^{(1)}=0\right\}
  \end{split}
\end{equation}
is called the \emph{Dirac kernel} of $S'$. The Dirac kernel plays a crucial role in our arguments.

A splitting of the above short exact sequence --- that is, a linear
map $\Sigma:V\to \odot^2 S'$ such that $\Sigma(v)^{(1)}=v$ for all
$v\in V$ --- is called a \emph{section} associated to $S'$. A section
associated to $S'$ always exists and it is unique up to
elements in the Dirac kernel.

\begin{definition}
  \label{def:Lie pair}
  Let $\fh_{(\varphi,S')}$ be the
  subspace of $\fso(V)$ given by
  \begin{equation}
	\label{isotropy-phi}
    \begin{split}
      \fh_{(\varphi,S')} &=\gamma^{\varphi}(\mathfrak{D})= \left\{\gamma^{\varphi}(\omega)~\middle|~\omega\in\odot^2
        S'~\text{with}~\omega^{(1)}=0\right\}.
    \end{split}
  \end{equation}
  The pair $(\varphi,S')$ is called a \emph{Lie pair} if
	\begin{equation}
	\label{eq:Lie-pair}
	\begin{aligned}
A\cdot\varphi&=0\;\\ 
 \sigma(A)S'&\subset S' 
  \end{aligned}
	\end{equation}
	for every $A\in\fh_{(\varphi,S')}$.
\end{definition}
The name ``Lie pair'' is motivated by the fact that the associated \eqref{isotropy-phi} is, in that case, a
  Lie subalgebra of $\fso(V)$.  
	The following result gives necessary conditions satisfied by
any tamed 
subdeformation. 
\begin{proposition}\cite{FOFS2017II}
  \label{prop:LiePair}
Let $\fg$ be a tamed filtered subdeformation,
  with underlying graded Lie algebra $\fa=\fh\oplus S'\oplus V$. Then the symbol $(\varphi,S')$ is a Lie pair and:
  \begin{enumerate}
  \item the isotropy algebra $\fh=\fh_{(\varphi,S')}$;
  \item the map $X:V\to\fso(V)$ is determined, up to elements in
    $\fh$, by the identity
    \begin{equation}
		\label{eq:map-X}
      X=\gamma^\varphi\circ\Sigma\;,
    \end{equation}
    where $\Sigma$ is any section associated to $S'$.
  \end{enumerate}  
  In particular $\fg$ is fully determined, up to isomorphism, by the associated symbol. 
\end{proposition}
\begin{corollary}
\label{corollary:geometric-symbol}
Any highly supersymmetric supergravity background $(M,g,F)$ is fully determined by its geometric symbol $\symb(M,g,F)$, up to local isometry.
\end{corollary}
\section{An alternative proof of Theorem \ref{thm:No-Go31}}
\label{sec:4}
Let $(M,g,F)$ be a background with {\it exactly} $31$ Killing spinors (locally), and $
(\varphi,S')$ the associated geometric symbol. Since $S'$ has dimension $31$, we may write
$S'=(\mathbb R s)^\perp$ as the symplectic orthogonal of the line spanned by some spinor $s\in S$ and,
as in the original proof of Theorem \ref{thm:No-Go31} in \cite{GGPR}, there are two cases to examine.

Indeed, as shown, e.g., in \cite{B}, there are two orbits of
$\Spin(V)$ on $\mathbb P(S)$, distinguished by the causal character of the associated Dirac current $\xi=\kappa(s,s)$: either
timelike or null. We give our proof of Theorem \ref{thm:No-Go31} using Dirac kernels only in the timelike case. Although the overall strategy in the null case is very close, the actual details are different and deserve a separate treatment. We leave them to the interested reader.

We first recall some known facts on realizations of the spinor module $S$ \cite{GP, GaGP} (see also e.g. \cite[Proposition 20]{FOFS2017II}, which follows the conventions used in this paper).
We normalise $s\in S$ so that $\eta(\xi,\xi)=1$ and consider the orthogonal decomposition $V=\mathbb R\xi\oplus W$. 
We set $\Omega^{(q)}=\omega^{(q)}(s,s)\in\Lambda^q V$ for $q=1,2,5$; by definition $\Omega^{(1)}=\xi$.
The endomorphism $J=\Omega^{(2)}\in\Lambda^2 V\cong\mathfrak{so}(V)$ acts trivially on $\xi$ and as a complex structure on $W$. Clearly the stabilizer
\begin{equation*}
\begin{aligned}
\stab_{\fso(V)}(s)&=\stab_{\fso(V)}(\Omega^{(1)})\cap\stab_{\fso(V)}(\Omega^{(2)})\cap\stab_{\fso(V)}(\Omega^{(5)})\\
&=\mathfrak u(W,J)\cap\stab_{\fso(V)}(\Omega^{(5)})\;,
\end{aligned}
\end{equation*}
and it turns out that $\Omega=\Omega^{(5)}+\tfrac12\Omega^{(1)}\wedge\Omega^{(2)}\wedge\Omega^{(2)}$ is the real component of a complex volume form $\Omega-i(\imath_\xi\star\Omega)$ on $W$. Ultimately $\stab_{\fso(V)}(s)=\mathfrak{su}(W,J)\cong\mathfrak{su}(5)$.

For our purposes, it is sufficient to work at the complexified level. In particular $\mathbb W=W\otimes\CC$ decomposes as the direct sum of isotropic subspaces $\mathbb W =
W^{10}\oplus W^{01}$, where $  W^{10}=\left<w^{10}=\tfrac{1}{2}(w-iJw)\;\middle |\; 
    w\in W\right>$ and $W^{01}=\overline{W^{10}}$
are  irreducible complex
$\mathfrak{su}(5)$-modules. 
The Dirac spinor module $\mathbb S=S\otimes \CC$ is then given by 
\begin{equation}
\label{Dirac-spinor-module}
\mathbb S=\Lambda^{(\bullet,0)}=\bigoplus_{0\leq p\leq 5}\Lambda^{(p,0)}_{[5-2p]}\;,
\end{equation}
where $\Lambda^{(p,0)}=\Lambda^p (W^{10})^*$ is the irreducible $\fsu(5)$-module of
the $(p,0)$-forms on $\mathbb W$ and the square brackets denote the charges of $\sigma(2J)$, with imaginary units removed.

As an aside, we recall that the modules $\Lambda^{(p,q)}$ of $(p,q)$-forms are $\fsu(5)$-irreducible only for $p=0$ or $q=0$ and that for all $0\leq p,q\leq 5$ we have natural isomorphisms
$$\Lambda^{(p,0)}\cong\Lambda^{(0,5-p)}\qquad\text{and}\qquad\Lambda^{(p,q)}\cong \Lambda^{(5-q,5-p)}\;.$$  
For example $\Lambda^{(1,1)}\cong\Lambda^{(1,1)}_o\oplus\Lambda^{(0,0)}$ and $\Lambda^{(4,3)}\cong\Lambda^{(2,1)}\cong\Lambda^{(2,1)}_o\oplus\Lambda^{(1,0)}$ into irreducible $\mathfrak{su}(5)$-modules. For any
$u\in W^{10}$ and $t\in\mathbb S$, the
identities
\begin{equation}
  \label{eq:spin1}
\begin{aligned}
u\cdot t&:=-\sqrt2\, \imath_{u}t~,\\
  \overline u\cdot t&:=\sqrt2\, \overline u^{\flat}\wedge t~,\\
\xi\cdot t&:=\vol_W\cdot t\;,
\end{aligned}
\end{equation}
give an  irreducible representation of the complex Clifford algebra
$\mathbb{C}l(11) \cong2\CC(32)$, where $\vol_W\in\Lambda^{10} W$ is the volume element of $W$.
The only trivial $\fsu(5)$-submodules of $\mathbb S$ are $\Lambda^{(0,0)}_{[5]}$ and $\Lambda^{(5,0)}_{[-5]}$, more precisely we have
\begin{equation*}
\label{eq:trivial-modules}
\Lambda^{(5,0)}_{[-5]}=\big\langle s^{(5,0)}=\tfrac12 (s+i\xi\cdot s)\big\rangle\quad\text{and}\quad \Lambda^{(0,0)}_{[5]}=\big\langle s^{(0,5)}=\tfrac12 (s-i\xi\cdot s)\big\rangle\;,
\end{equation*}
whence the two elements $s=s^{(5,0)}+s^{(0,5)}$ and $\xi\cdot s=-i(s^{(5,0)}-s^{(0,5)})$ do {\it not} have a definite charge. 
The proof of the following lemma is straightforward.
\begin{lemma}
The space of bispinors $\odot^2 \mathbb S$ decomposes into $\fsu(5)$-modules as
\begin{equation}
\label{eq:decomposition-forms-su5-I}
\begin{aligned}
\odot^2 \mathbb S&\cong \Lambda^1 \V\oplus\Lambda^2 \V\oplus \Lambda^5 \V\\
&\cong \big(\Lambda^1 \W\oplus\Lambda^0 \W\big)\oplus\big(\Lambda^2 \W\oplus\Lambda^1 \W\big)\oplus\big(\Lambda^5 \W\oplus\Lambda^4 \W\big)
\end{aligned}
\end{equation}
with
\begin{equation}
\label{eq:decomposition-forms-su5-II}
\begin{aligned}
\Lambda^0 \W&=\CC\Omega^{(1)}\\
\Lambda^1 \W&=\Lambda^{(1,0)}_{[-2]}\oplus\Lambda^{(0,1)}_{[2]}\\
\Lambda^2 \W&=\Lambda^{(2,0)}_{[-4]}\oplus \Lambda^{(0,2)}_{[4]}\oplus{\color{red} \Lambda^{(1,1)}_{o[0]}}\oplus \CC\Omega^{(2)}\\
\Lambda^4 \W&=\Lambda^{(4,0)}_{[-8]}\oplus \Lambda^{(0,4)}_{[8]}\oplus
{\color{red} \Lambda^{(3,1)}_{o[-4]}}\oplus{\color{red} \Lambda^{(1,3)}_{o[4]}}
\oplus\Lambda^{(2,0)}_{[-4]}\oplus \Lambda^{(0,2)}_{[4]}\\
&\;\;\;\oplus {\color{red} \Lambda^{(2,2)}_{o[0]}}\oplus{\color{red} \Lambda^{(1,1)}_{o[0]}}\oplus \mathbb C(\Omega^{(1)}\wedge\Omega^{(2)}\wedge\Omega^{(2)})
\\
\Lambda^5 \W&=\CC\Omega\oplus\CC(\imath_\xi\star\Omega)
\oplus {\color{red} \Lambda^{(4,1)}_{o[-6]}}\oplus{\color{red} \Lambda^{(1,4)}_{o[6]}}\oplus
\Lambda^{(3,0)}_{[-6]}\oplus \Lambda^{(0,3)}_{[6]}\\
&\;\;\;\oplus{\color{red} \Lambda^{(3,2)}_{o[-2]}}\oplus{\color{red} \Lambda^{(2,3)}_{o[2]}}\oplus
{\color{red} \Lambda^{(2,1)}_{o[-2]}}\oplus {\color{red} \Lambda^{(1,2)}_{o[2]}}\oplus
\Lambda^{(1,0)}_{[-2]}\oplus\Lambda^{(0,1)}_{[2]}
\end{aligned}
\end{equation}
into $\mathfrak{su}(5)$-irreducibles. The modules not isomorphic to $\Lambda^{(p,0)}$ for some $0\leq p\leq 5$ are colored red.
The trivial $\fsu(5)$-modules generated by $\Omega^{(1)}$, $\Omega^{(2)}$ and 
$\Omega^{(1)}\wedge\Omega^{(2)}\wedge\Omega^{(2)}$ have charge zero, 
whereas $\Omega$ together with $\imath_\xi\star\Omega$ generate the sum $\Lambda^{(5,0)}_{[-10]}\oplus \Lambda^{(0,5)}_{[10]}$ of the trivial $\fsu(5)$-modules with charge $\pm 10$.

In particular, the isotypic decomposition of $\odot^2 \mathbb S$ is
\begin{equation*}
\label{isotypic-decomposition-odot2S}
\begin{aligned}
\odot^2 \mathbb S&\cong 5\Lambda^{(0,0)}\oplus 4\left(\Lambda^{(1,0)}\oplus\Lambda^{(0,1)}\right)\oplus 3\left(\Lambda^{(2,0)}\oplus \Lambda^{(0,2)}\right)\oplus 2\Lambda^{(1,1)}_{o}\oplus \Lambda^{(2,2)}_{o}\\
&\;\;\;\oplus\left(\Lambda^{(3,1)}_o\oplus \Lambda^{(1,3)}_o\right)\oplus\left(\Lambda^{(4,1)}_o\oplus\Lambda^{(1,4)}_o\right)
\oplus\left(\Lambda^{(3,2)}_{o}\oplus\Lambda^{(2,3)}_{o}\right)
\oplus\left(\Lambda^{(2,1)}_o\oplus\Lambda^{(1,2)}_o\right),
\end{aligned}
\end{equation*}
where we have grouped conjugated modules with brackets.
\end{lemma}

Note that $\mathbb S=\big\langle s,\xi\cdot s\big\rangle\oplus \mathbb{\widetilde S}$ decomposes as the direct sum of 
orthogonal $\fsu(5)$-modules
\begin{equation}
\label{eq:ortoh-decomp}
\big\langle s,\xi\cdot s\big\rangle=\Lambda^{(5,0)}_{[-5]}\oplus\Lambda^{(0,0)}_{[5]}\quad\text{and}\quad\widetilde{\mathbb{S}}=\Lambda^{(4,0)}_{[-3]}\oplus\Lambda^{(3,0)}_{[-1]}\oplus\Lambda^{(2,0)}_{[1]}\oplus\Lambda^{(1,0)}_{[3]}\;.
\end{equation}
We would like to estimate the size of the Dirac kernel of 
$\mathbb S'=(\mathbb C s)^\perp=\mathbb C s \oplus\widetilde{\mathbb{S}}$:
\begin{proposition}
\label{Dirac-kernel-properties}
The Dirac kernel $\mathfrak D$ of $\mathbb S'$ includes at least the following elements:
\begin{enumerate}	
\item Any $\omega\in\Lambda^{(2,0)}\oplus \Lambda^{(0,2)}\subset\Lambda^2 \W$, modulo $\Lambda^4 \W\subset \Lambda^5\V$;
\item Any $\omega\in\Lambda^{(4,1)}\oplus\Lambda^{(1,4)}\subset\Lambda^5\W$, modulo $\Lambda^2\W\oplus\Lambda^4\W\subset\Lambda^2\V\oplus\Lambda^5\V$;
	\item Any $\omega\in\Lambda^{(3,2)}\oplus\Lambda^{(2,3)}\subset\Lambda^5\W$, modulo $\Lambda^4\W\subset\Lambda^5\V$;
\end{enumerate}
\end{proposition}
\begin{proof}
First of all, we note that
\begin{equation*}
\begin{aligned}
\odot^2 \mathbb S&\cong \odot^2 \mathbb S'\oplus\big(\mathbb S'\odot(\xi\cdot s)\big)\oplus\mathbb C\big(\xi\cdot s\odot\xi\cdot s\big)\\
\odot^2 \mathbb S'&\cong \odot^2 \widetilde{\mathbb S}\oplus\big(\widetilde{\mathbb S}\odot s\big)\oplus\CC\big(s\odot s\big)
\end{aligned}
\end{equation*}
into $\fsu(5)$-modules, and a routine computation using \eqref{eq:ortoh-decomp} gives
\begin{equation}
\label{eq:big-part-of-it}
\begin{aligned}
\odot^2 \widetilde{\mathbb S}&\cong 2\Lambda^{(0,0)}_{[0]}\oplus 2\left(\Lambda^{(1,0)}_{[-2]}\oplus\Lambda^{(0,1)}_{[2]}\right)\oplus \left(\Lambda^{(2,0)}_{[-4]}\oplus \Lambda^{(0,2)}_{[4]}\right)\oplus 2\Lambda^{(1,1)}_{o[0]}\oplus \Lambda^{(2,2)}_{o[0]}\\
&\oplus\left(\Lambda^{(3,1)}_{o[-4]}\oplus \Lambda^{(1,3)}_{o[4]}\right)\oplus\left(\Lambda^{(4,1)}_{o[-6]}\oplus\Lambda^{(1,4)}_{o[6]}\right)
\oplus\left(\Lambda^{(3,2)}_{o[-2]}\oplus\Lambda^{(2,3)}_{o[2]}\right)\\
&
\oplus\left(\Lambda^{(2,1)}_{o[-2]}\oplus\Lambda^{(1,2)}_{o[2]}\right)
\end{aligned}
\end{equation}
and $\widetilde{\mathbb S}\odot s\cong \widetilde{\mathbb S}\cong \left(\Lambda^{(1,0)}\oplus\Lambda^{(0,1)}\right)\oplus\left(
\Lambda^{(2,0)}\oplus\Lambda^{(0,2)}\right)$ into $\fsu(5)$-irreducible modules. The latter ones do not have a definite charge, since $s=s^{(5,0)}+s^{(0,5)}\in \Lambda^{(5,0)}_{[-5]}\oplus\Lambda^{(0,0)}_{[5]}$.
 
Comparing \eqref{eq:big-part-of-it} with the isotypic decomposition of $\odot^2 \mathbb S$ says that all the $\fsu(5)$-submodules 
colored in red in \eqref{eq:decomposition-forms-su5-II} are contained in
$\odot^2 \widetilde{\mathbb S}$. Clearly they also lie in $\mathfrak D$.
It remains to deal with the isotypic components of type $\Lambda^{(p,0)}$, where $0\leq p\leq 4$.
\vskip0.05cm\par\noindent
\underline{\it Claim 1}
Note that the two copies of $\left(\Lambda^{(2,0)}\oplus\Lambda^{(0,2)}\right)$ with charge $\pm 4$ are contained in $\Lambda^2\W$ and $\Lambda^4\W$. Taking an appropriate $t=t^{(4,0)}+t^{(3,0)}\in\Lambda^{(4,0)}_{[-3]}\oplus\Lambda^{(3,0)}_{[-1]}\subset  \widetilde{\mathbb{S}}$ 
and $u_1,u_2\in W^{10}$, we may use the first equation in \eqref{eq:spin1} to get
\begin{equation*}
\begin{aligned}
\langle t,(u_1\wedge u_2)\cdot t\rangle
&=2\langle t^{(4,0)},\underbrace{(u_1\wedge u_2)\cdot t^{(3,0)}}_{\text{element of}\;\Lambda^{(1,0)}_{[3]}}\rangle\neq 0\;.
\end{aligned}
\end{equation*}
\vskip-0.1cm\par\noindent
It follows that $\odot^2 \widetilde{\mathbb S}$ projects surjectively to $\Lambda^{(2,0)}\subset\Lambda^2 \W$
modulo elements in $\Lambda^4\W$, whence to $\Lambda^{(2,0)}\oplus\Lambda^{(0,2)}$ as well (because $\widetilde{\mathbb S}$ is a module of real type).
\vskip0.1cm\par\noindent
\underline{\it Claim 2} 
We already established that $\Lambda^{(4,1)}_{o[-6]}\oplus\Lambda^{(1,4)}_{o[6]}\subset \mathfrak D$. Now
$
\Lambda^{(3,0)}_{[-6]}\oplus \Lambda^{(0,3)}_{[6]}
$
is the only module of $\odot^2 \mathbb S$ isomorphic to $\Lambda^{(2,0)}\oplus \Lambda^{(0,2)}$ and with charge different from $\pm 4$. It follows that $\widetilde{\mathbb S}\odot s$ projects surjectively to it, modulo elements in $\Lambda^2\W\oplus\Lambda^4\W$. 
\vskip0.1cm\par\noindent
\underline{\it Claim 3} 
We already established that the submodules $\Lambda^{(3,2)}_{o[-2]}\oplus\Lambda^{(2,3)}_{o[2]}\oplus
\Lambda^{(2,1)}_{o[-2]}\oplus \Lambda^{(1,2)}_{o[2]}$ of $\Lambda^5\W$ belong to $\mathfrak D$, so it remains to consider
$\Lambda^{(1,0)}_{[-2]}\oplus\Lambda^{(0,1)}_{[2]}\subset \Lambda^5\W$. 

There are four modules isomorphic to $\Lambda^{(1,0)}$ in $\odot^2 \mathbb S$: two copies of $\Lambda^{(1,0)}_{[-2]}\subset \odot^2 \widetilde{\mathbb S}$ and two other two copies
\begin{equation*}
\label{eq:almost-the-end}
\begin{aligned}
s\odot\Lambda^{(1,0)}_{[3]}&=\left\{s^{(5,0)}\odot\alpha+s^{(0,5)}\odot\alpha\mid\alpha\in\Lambda^{(1,0)}_{[3]}\right\}\subset s\odot\widetilde{\mathbb{S}}\;,\\
(\xi\cdot s)\odot\Lambda^{(1,0)}_{[3]}&=\left\{s^{(5,0)}\odot\alpha-s^{(0,5)}\odot\alpha\mid\alpha\in\Lambda^{(1,0)}_{[3]}\right\}\subset (\xi\cdot s)\odot\widetilde{\mathbb{S}}
\end{aligned}
\end{equation*} 
that do not have a definite charge. It follows that $2\Lambda^{(1,0)}_{[-2]}\oplus \big(s\odot\Lambda^{(1,0)}_{[3]}\big)\subset\odot^2 \mathbb S'$ projects surjectively to $3\Lambda^{(1,0)}_{[-2]}\subset \odot^2\mathbb S$, modulo elements of charge $8$.
A similar property holds for $3\Lambda^{(0,1)}_{[2]}$ and this readily implies our last claim.
\end{proof}
Let $(\varphi,S')$ be the geometric symbol of a background with exactly $31$ Killing spinors.
By Proposition \ref{prop:LiePair} the isotropy $\fh=\fh_{(\varphi,S')}=\gamma^{\varphi}(\mathfrak{D})$ and the symbol $(\varphi,S')$ is a Lie pair, in particular $\fh\subset\stab_{\fso(V)}(S')=\stab_{\fso(V)}(\mathbb R s)^\perp=\stab_{\fso(V)}(\mathbb Rs)$. Since 
$\xi$ is timelike, we actually have $$\fh\subset\stab_{\fso(V)}(s)=\fsu(5)\subset\stab_{\fso(V)}(\xi)\;.$$
The idea is very simple: we write $\varphi=\xi\wedge\phi+\Phi$, for some $\phi\in\Lambda^3 W$ and $\Phi\in\Lambda^4 W$ and use the definition \eqref{definition-gamma-useful} of the map $\gamma^\varphi$ with $w=\xi$ and $\omega\in\mathfrak D$:
\begin{equation}
\label{eq:crucial-equation}
\begin{aligned}
0=\eta(\gamma^\varphi(\omega)v,\xi)
&=\tfrac{1}{3}\eta\big(\imath_v\imath_\xi\varphi,\omega^{(2)}\big)
+\tfrac{1}{6}\eta\big(\imath_v\imath_\xi\star\varphi,\omega^{(5)}\big)\\
&=\tfrac{1}{3}\eta\big(\!\!\!\!\!\!\!\!\!\!\underbrace{\imath_v\phi}_{\text{element of}\;\Lambda^2 W}\!\!\!\!\!\!\!\!\!\!\;,\omega^{(2)}\big)
+\tfrac{1}{6}\eta\big(\!\!\!\!\!\!\!\underbrace{\imath_v\bigstar\Phi}_{\text{element of}\;\Lambda^5 W}\!\!\!\!\!\!\!\!\!\!\;\;,\omega^{(5)}
\big)\;,
\end{aligned}
\end{equation}
for all $v\in W$, with $\bigstar$ the Hodge star on $W$. We now exploit that $\mathfrak D$ is sufficiently big.

Taking $\omega$ as in (1) of Proposition \ref{Dirac-kernel-properties}, equation \eqref{eq:crucial-equation} becomes
$
\eta(\imath_v\phi,\omega^{(2)})=0$
for all $\omega^{(2)}\in\Lambda^{(2,0)}\oplus \Lambda^{(0,2)}\subset\Lambda^2\W$. Decomposing $\phi=\phi^{(3,0)}+\phi^{(2,1)}+\phi^{(1,2)}+\phi^{(0,3)}$ into types   and taking $v\in W^{10}$ gives $\phi^{(3,0)}=\phi^{(1,2)}=0$. Similarly $\phi^{(2,1)}=\phi^{(0,3)}=0$, so $\phi=0$.
If $\omega$ is now as in (2) of Proposition \ref{Dirac-kernel-properties}, then
$\eta(\imath_v\bigstar\Phi,\omega^{(5)})=0$
for all $\omega^{(5)}\in \Lambda^{(4,1)}\oplus\Lambda^{(1,4)}$.
Decomposing $\bigstar\Phi\in\Lambda^6W$ into types yields $(\bigstar\Phi)^{(p,q)}=0$ except for
$(\bigstar\Phi)^{(3,3)}$, but this is zero by (3) of Proposition \ref{Dirac-kernel-properties} and a similar argument.

Hence $\varphi=0$, so $(M,g)$ is locally isometric to Minkowski spacetime, in particular 
it is maximally supersymmetric, a contradiction. 
The proof is completed. 
\section{The reconstruction problem for supergravity backgrounds}
\label{sec:class-probl-highly}
It is not true that every Lie pair has a corresponding filtered subdeformation, in other words, it is the geometric symbol of a background. Indeed,
the Lie brackets of a tamed filtered subdeformation $\fg$ are
given by
\begin{equation*}  
  \begin{aligned}[m]
    [A,B]&=AB-BA\\
    [A,s]&=\sigma(A)s\\
    [A,v]&=Av+[A, X_v] - X_{Av}
  \end{aligned}
  \;\;
  \begin{aligned}[m]
    [s,s]&=\kappa(s,s) + \gamma^{\varphi}(s,s) - X_{\kappa(s,s)}\\
    [v,s]&=\beta^{\varphi}(v,s) + \sigma(X_v)s\\
    [v,w]&=X_v w - X_w v + [X_v,X_w] - X_{X_v w - X_w v} + R(v,w)
  \end{aligned}
\end{equation*}
for all $A,B\in\fh$, $v,w\in V$, $s\in S'$. 
Here $\fh=\fh_{(\varphi,S')}$ and $X:V\to\fso(V)$ is as in \eqref{eq:map-X}.
The rest of the data $R : \Lambda^2 V \to \fso(V)$ also depends on the Lie pair, as we now recall.

First of all, it bears reminding
that the right-hand sides of the above Lie brackets 
take values in $\fa=\fh \oplus S' \oplus V$, but the individual
terms may not.  
Explicitly, 
\begin{align}
\label{eq:set-theoretic-I}
\beta^{\varphi}(v,s) + \sigma(X_v)s&\in S'\\
\label{eq:set-theoretic-II}
[A,X_v] - X_{Av} &\in \fh\\
\label{eq:set-theoretic-III}
\gamma^{\varphi}(s,s) - X_{\kappa(s,s)}&\in\fh\\
\label{eq:set-theoretic-IV}
[X_v,X_w] - X_{X_v w - X_w v} + R(v,w)&\in\fh
\end{align}
for all $A\in\fh$, $v,w\in V$, $s\in S'$. Moreover the Lie brackets are subject to the Jacobi identities. There are ten components of them, and five are simply satisfied by equivariance and the fact that $\fh\subset\stab_{\fso(V)}(S')\cap\stab_{\fso(V)}(\varphi)$ \cite[\S 4.1]{FOFS2017II}. The Jacobi identity with three odd elements is the cocycle condition \eqref{eq:secondcocyclecondition}.

The remaining four components of the Jacobi identities are more involved:
\begin{itemize}
	\item The $[\fh V V]$ Jacobi identity is satisfied if and only if
\begin{equation}
  \label{eq:Requiv}
 R : \Lambda^2 V \to \fso(V) \quad\text{is $\fh$-equivariant};
\end{equation}
\item The $[S'S'V]$ Jacobi is
equivalent to
\begin{equation}
  \label{eq:ssvJac}
  \begin{split}
    \tfrac12 R(v,\kappa(s,s))w &=\kappa((X_v\beta^{\varphi})(w,s),s)+
    \gamma^\varphi(\beta^{\varphi}_{v}(s),s)w\\
		&=\kappa((X_v\beta^{\varphi})(w,s),s)
	-\kappa(\beta^\varphi_v(s) , \beta^\varphi_w(s)) -
    \kappa(\beta^\varphi_w \beta^\varphi_v(s),s)
  \end{split}
    \end{equation}
for all $s \in S'$, $v \in V$ and $w \in V$;
\item The $[S'VV]$ Jacobi identity expands to the following condition
\begin{equation}
  \label{eq:svvJac}
  R(v,w)s = (X_v\beta^{\varphi})(w,s) - (X_w\beta^{\varphi})(v,s) +
  [\beta^\varphi_v,\beta^\varphi_w](s),
\end{equation}
for all $s \in S'$ and $v,w \in V$; 
\item Finally, the $[VVV]$ Jacobi identity expands to the algebraic and differential Bianchi identities
\begin{align}  
\label{eq:Bianchi-I}
R(u,v)w + R(v,w) u + R(w,u) v &= 0\;,\\
\label{eq:Bianchi-II}
(X_u R)(v,w) + (X_v R)(w,u) + (X_w R)(u,v) &= 0\;,
\end{align}
for all $u,v,w \in V$. 
\end{itemize}
Using the local Homogeneity Theorem, it is not difficult to see that each of equations \eqref{eq:ssvJac}-\eqref{eq:svvJac}
determines uniquely  the curvature tensor in terms of the Lie pair.

In summary, if we aim to construct backgrounds via filtered subdeformations and Theorem \ref{thm:reconstruction}, a long task awaits us: find a Lie pair $(\varphi,S')$, compute the isotropy $\fh$ and the map $X$ via
the Dirac kernel, check that equations \eqref{eq:set-theoretic-I}-\eqref{eq:Bianchi-II} are satisfied, for a putative curvature tensor to be determined, and don't forget also $d\varphi=0$. 

The following result drastically simplifies the situation:
\begin{theorem}
\label{thm:main-reduction}
Assume $(\varphi,S')$ is a Lie pair satisfying \eqref{eq:set-theoretic-I} and there is $R:\Lambda^2 V \to \fso(V)$ so that
\eqref{eq:ssvJac} and \eqref{eq:svvJac} hold. Then the identities \eqref{eq:set-theoretic-II}-\eqref{eq:Requiv} and \eqref{eq:Bianchi-I}-\eqref{eq:Bianchi-II} automatically hold. 
\end{theorem}
\begin{proof}
We split the proof in six steps, one for each identity.
\vskip0.05cm\par\noindent
\underline{\it Step I}
For all $A\in\fh$, $v\in V$, we have
\begin{equation*}
\begin{aligned}
{}[A,X_v]-X_{Av}&=[A,\gamma^\varphi(\Sigma v)]-\gamma^\varphi(\Sigma Av)=\gamma^\varphi(A\Sigma v-\Sigma Av)
\end{aligned}
\end{equation*}
and $A\Sigma v-\Sigma Av\in  \mathfrak{D}$, since
$
(A\Sigma v)^{(1)}=A(\Sigma v)^{(1)}=Av=(\Sigma Av)^{(1)}
$ and $\Sigma:V\to\odot^2 S'$. Identity \eqref{eq:set-theoretic-II} follows then directly from $\fh=\fh_{(\varphi,S')}=\gamma^{\varphi}(\mathfrak{D})$. 
\vskip0.105cm\par\noindent
\underline{\it Step II} In a similar way
$\gamma^{\varphi}(\omega) - X_{\omega^{(1)}}
=\gamma^{\varphi}(\omega -\Sigma\omega^{(1)})
$
is an element of $\fh$ for all $\omega\in\odot^2 S'$, proving \eqref{eq:set-theoretic-III}. 
\vskip0.05cm\par\noindent
\underline{\it Step III}
We now establish \eqref{eq:set-theoretic-IV}. Using \eqref{eq:ssvJac} and the definition of $\gamma^\varphi$ we compute
\begin{equation*}
\begin{split}
  R(v,\kappa(s,s))w &=2\kappa\big((X_v\beta^{\varphi})(w,s),s\big)+
    2\gamma^\varphi\big(\beta^{\varphi}_{v}(s),s\big)w
		\\
		&=2\kappa\big(\sigma(X_v)(\beta^\varphi(w,s)),s\big)
		+\gamma^{\varphi}(s,s)X_v w
		-2\kappa\big(\beta^\varphi(w,\sigma(X_v)s),s\big)\\
		&\;\;\;+2\gamma^\varphi\big(\beta^{\varphi}_{v}(s),s\big)w\\
		&=-X_v\big(\gamma^{\varphi}(s,s)w\big)
	-2\kappa\big(\beta^\varphi(w,s),\sigma(X_v)s\big)
		-2\kappa\big(\beta^\varphi(w,\sigma(X_v)s),s\big)\\
		&\;\;\;+\gamma^{\varphi}(s,s)X_v w+2\gamma^\varphi\big(\beta^{\varphi}_{v}(s),s\big)w\\
	&=-X_v\big(\gamma^{\varphi}(s,s)w\big)+2\gamma^{\varphi}\big(s,\sigma(X_v)s\big)w+\gamma^{\varphi}(s,s)X_v w\\
	&\;\;\;+2\gamma^\varphi\big(\beta^{\varphi}_{v}(s),s\big)w\\
		&=[\gamma^{\varphi}(s,s),X_v]w+2\gamma^{\varphi}\big(s,\underbrace{(\beta^{\varphi}_{v}(s)+\sigma(X_v)s)}_{\text{element of}\;S'\;\text{(by}\;\eqref{eq:set-theoretic-I})}\!\big)w
	\end{split}
    \end{equation*}
for all $s\in S'$, $v,w\in V$. On the other hand
\begin{equation*}
\begin{aligned}
{}[X_v,X_{\kappa(s,s)}]&=[X_v,\gamma^\varphi(\Sigma\kappa(s,s))]
\end{aligned}
\end{equation*}
and
\begin{equation*}
\begin{aligned}
X_{X_{\kappa(s,s)} v}-X_{X_v \kappa(s,s)}&=X_{\gamma^{\varphi}(\Sigma\kappa(s,s))v}-X_{X_v \kappa(s,s)}-X_{\gamma^{\varphi}(s,s)v}+X_{\gamma^{\varphi}(s,s)v}\;,
\end{aligned}
\end{equation*}
for all $s\in S'$ and $v\in V$. 

We now sum up the three contributions to the identity \eqref{eq:set-theoretic-IV} and regroup the various terms into the sum of
\begin{equation}
\label{eq:final-IIIA}
\begin{split}
&{}[\gamma^{\varphi}(s,s)-\gamma^\varphi(\Sigma\kappa(s,s)),X_v]+X_{\gamma^{\varphi}(\Sigma\kappa(s,s))v}-X_{\gamma^{\varphi}(s,s)v}
\end{split}
\end{equation}
and
\begin{equation}
\label{eq:final-IIIB}
\begin{split}
&2\gamma^{\varphi}\big(s,\beta^{\varphi}_{v}(s)+\sigma(X_v)s\big)+X_{\gamma^{\varphi}(s,s)v}-X_{X_v \kappa(s,s)}\;,
\end{split}
\end{equation}
where $\gamma^{\varphi}(s,s)-\gamma^\varphi(\Sigma\kappa(s,s))\in\fh=\gamma^{\varphi}(\mathfrak{D})$.
The term \eqref{eq:final-IIIA} belongs to $\fh$ due to \eqref{eq:set-theoretic-II}, which we established in step I.
Finally, the Dirac current
\begin{equation*}
\kappa(s,\beta^{\varphi}_{v}(s)+\sigma(X_v)s)=-\tfrac{1}{2}\gamma^\varphi(s,s)v+\tfrac{1}{2}X_v(\kappa(s,s))\;,
\end{equation*}
so the term \eqref{eq:final-IIIB} is also in $\fh$, thanks to the identity \eqref{eq:set-theoretic-III} established in step II.
\vskip0.05cm\par\noindent
\underline{\it Step IV}
We prove that $R$ is $\fh$-equivariant. A direct computation using \eqref{eq:ssvJac} yields
\begin{equation*}
  \begin{split}
    \tfrac12 (A\cdot R)(v,\kappa(s,s))w &=\kappa\big((A\cdot X)_v\beta^{\varphi})(w,s),s\big)
		=\kappa\big(\beta^{(A\cdot X)_v\cdot\varphi}(w,s),s\big)
		=0\;,
		\end{split}
\end{equation*}
for all $A\in\fh$, $s\in S'$ and $v,w\in V$. Here we used that
$(A\cdot X)_v=[A,X_v]-X_{Av}$ is an element of $\fh$ by step I, hence it annihilates $\varphi$.
\vskip0.05cm\par\noindent
\underline{\it Step V} It is sufficient to establish \eqref{eq:Bianchi-I} with $w=\kappa(s,s)$ for all $s\in S'$. 
Now
\begin{equation*}
\begin{aligned}
\tfrac12\big(R(v,\kappa(s,s))u+R(\kappa(s,s),u)v\big)&=\kappa((X_v\beta^{\varphi})(u,s),s)-\kappa((X_u\beta^{\varphi})(v,s),s)\\
	&\;\;\;-\kappa([\beta^\varphi_u, \beta^\varphi_v](s),s)
	\end{aligned}
	\end{equation*}
by equation \eqref{eq:ssvJac} and this term is also equal to $\tfrac12 R(v,u)\kappa(s,s)=\kappa(R(v,u)s,s)$ by \eqref{eq:svvJac}.
\vskip0.05cm\par\noindent
\underline{\it Step VI} 
The last step is the most involved and needs a  preliminary crucial observation: 
it is enough to establish \eqref{eq:Bianchi-II} when one of the elements of $V$ is of the form $\omega^{(1)}$ for some $\omega\in\odot^2 S'$ in the image of the section $\Sigma:V\to \odot^2 S'$. For simplicity of exposition, we denote $\omega=s\odot s$ with $s\in S'$, although it is really a sum of decomposable bispinors. Our assumption on $\omega$ reads then 
$
s\odot s=\omega=\Sigma\omega^{(1)}=\Sigma(\kappa(s,s))\;,
$
whence
\begin{equation}
\label{eq:crucial-assumption}
X_{\kappa(s,s)}=\gamma^{\varphi}\big(\Sigma(\kappa(s,s))\big)=\gamma^{\varphi}(s,s)\;.
\end{equation}
We will crucially use this property in the proof.

Now $A\cdot\beta^\varphi=\beta^{A\cdot\varphi}$ for all $A\in\fso(V)$ by $\fso(V)$-equivariance. A direct computation using this fact and \eqref{eq:ssvJac} yields
\begin{equation}
  \label{eq:ssvJac-infinitesimal}
  \begin{split}
    \tfrac12 (A\cdot R)(v,\kappa(s,s))w &=\kappa(\beta^{(A\cdot X_v\cdot\varphi)}(w,s),s)
		-\kappa(\beta^{(X_{Av}\cdot\varphi)}(w,s),s)\\
		&
	\;\;\;+\gamma^{A\cdot \varphi}(\beta^{\varphi}_{v}(s),s)w+\gamma^{\varphi}(\beta^{(A\cdot \varphi)}_{v}(s),s)w\;,
		\end{split}
    \end{equation}
		for all $A\in\fso(V)$. Applying \eqref{eq:ssvJac-infinitesimal} with $A=X_u$ and skew-symmetrizing in $u$ and $v$ says that the contribution
	\begin{equation}	
	\tfrac12 (X_u\cdot R)(v,\kappa(s,s))w-\tfrac12 (X_v\cdot R)(u,\kappa(s,s))w
	\end{equation}
		to the differential Bianchi identity is given by
		\begin{equation}
		  \label{eq:ssvJac-infinitesimal-skewsymmetrized}
\begin{aligned}
\kappa\big(\beta^{([X_u,X_v]\cdot\varphi)}(w,s),s\big)
		&-\kappa\big(\beta^{(X_{X_u v-X_v u}\cdot\varphi)}(w,s),s\big)
	+\gamma^{\varphi}\big(\beta^{(X_u\cdot \varphi)}_{v}(s),s\big)w\\
	&\!\!\!\!\!\!\!\!\!\!\!\!\!\!\!\!\!\!\!\!-\gamma^{\varphi}\big(\beta^{(X_v\cdot \varphi)}_{u}(s),s\big)w
+\gamma^{X_u\cdot \varphi}\big(\beta^{\varphi}_{v}(s),s\big)w
-\gamma^{X_v\cdot \varphi}\big(\beta^{\varphi}_{u}(s),s\big)w\;.
\end{aligned}
		\end{equation}
Using identity \eqref{eq:set-theoretic-IV}, established in step III, and the definition of a Lie pair, we directly see that the first two terms in \eqref{eq:ssvJac-infinitesimal-skewsymmetrized} are equal to
\begin{equation}
\label{eq:useful-I}
-\kappa\big(\beta^{(R(u,v)\cdot\varphi)}(w,s),s\big)=\tfrac12\gamma^{(R(u,v)\cdot\varphi)}(s,s) w\;.
\end{equation}	

We recall our crucial assumption \eqref{eq:crucial-assumption} and turn to compute the last contribution to the differential Bianchi identity:
\begin{equation}
\label{eq:contribution-II}
\begin{aligned}
\!\!\!\!\!\!\!\tfrac12(X_{\kappa(s,s)}\cdot R)(u,v) &=\tfrac12(\gamma^{\varphi}(s,s)\cdot R)(u,v) \\
&=-\tfrac12[R(u,v),\gamma^{\varphi}(s,s)]-\tfrac 12 R\big(\gamma^{\varphi}(s,s) u,v\big)-\tfrac 12 R\big(u,\gamma^{\varphi}(s,s)v\big)\\
&=-\tfrac12 \gamma^{( R(u,v)\cdot\varphi)}(s,s)-\gamma^{\varphi}\big(R(u,v) s,s\big)-\tfrac 12 R\big(\gamma^{\varphi}(s,s) u,v\big)\\
&\;\;\;\,-\tfrac 12 R\big(u,\gamma^{\varphi}(s,s)v\big)\;,
\end{aligned}
\end{equation}
where we used that $A\cdot\gamma^\varphi=\gamma^{A\cdot\varphi}$ for all $A\in\fso(V)$. Thanks to \eqref{eq:svvJac}
we may expand one of the terms in the next-to-last line:
\begin{equation}
\label{eq:useful-II}
\begin{aligned}
-\gamma^{\varphi}\big(R(u,v) s,s\big)&=-\gamma^{\varphi}\big(\beta^{(X_u\cdot\varphi)}(v,s),s\big)+\gamma^\varphi\big(\beta^{(X_v\cdot\varphi)}(u,s),s\big)\\
&\;\;\;\,-\gamma^{\varphi}\big([\beta^\varphi_u,\beta^\varphi_v](s),s\big)\;.
\end{aligned}
\end{equation}

Let us collect what we obtained so far: summing up \eqref{eq:ssvJac-infinitesimal-skewsymmetrized} and \eqref{eq:contribution-II}, and using \eqref{eq:useful-I} and \eqref{eq:useful-II}, we are left with
		\begin{equation}
		  \label{eq:last-term}
\begin{aligned}
	\gamma^{X_u\cdot \varphi}\big(\beta^{\varphi}_{v}(s),s\big)w&
-\gamma^{X_v\cdot \varphi}\big(\beta^{\varphi}_{u}(s),s\big)w-\gamma^{\varphi}\big([\beta^\varphi_u,\beta^\varphi_v](s),s\big)\\
&-\tfrac 12 R\big(\gamma^{\varphi}(s,s) u,v\big)w
-\tfrac 12 R\big(u,\gamma^{\varphi}(s,s)v\big)w\;.
\end{aligned}
		\end{equation}
Using the definition of $\gamma^{\varphi}$, equation \eqref{eq:ssvJac} and a direct computation, we see that the second line of 
\eqref{eq:last-term} is equal to
\begin{equation*}
  \begin{aligned}
R\big(\kappa(\beta^\varphi_u(s),s),v\big)w
-R\big(\kappa(\beta^\varphi_v(s),s),u\big)w
&=-\gamma^{X_u\cdot \varphi}\big(\beta^{\varphi}_{v}(s),s\big)w
+\gamma^{X_v\cdot \varphi}\big(\beta^{\varphi}_{u}(s),s\big)w\\
&\;\;\;\,+\gamma^{\varphi}\big([\beta^\varphi_u,\beta^\varphi_v](s),s\big)\;,
\end{aligned}
		\end{equation*}
so \eqref{eq:last-term} vanishes and the proof is completed.
\end{proof}
Motivated by this result, we give the following definition.
\begin{definition}
\label{def:main}
Let $(\varphi,S')$ be a Lie pair, that is, $S'$ is a subspace of $S$ with $\dim S'>16$, $\varphi\in\Lambda^4 V$ and 
$\fh=\fh_{(\varphi,S')}\subset \stab_{\fso(V)}(S')\cap\stab_{\fso(V)}(\varphi)$ (see Definition \ref{def:Lie pair} for details).
Let $\Sigma:V\to \odot^2 S'$ be any section associated with $S'$ and
set $X=\gamma^\varphi\circ\Sigma:V\to\fso(V)$ and $\alpha(v,w)=X_vw-W_w v$ for all $v,w\in V$. Then $(\varphi,S')$ is called an {\it abstract symbol} if 
\begin{align}
\label{eq:symbolII}
\beta^{\varphi}(v,s) + \sigma(X_v)s&\in S'\\
\label{eq:symbolIII}
d\varphi(v_0,\ldots,v_4) &= \sum_{i<j}(-1)^{i+j}
      \varphi(\alpha(v_i,v_j), v_0, \ldots, \hat v_i,\ldots,\hat
      v_j,\ldots,v_4) = 0\\
\label{eq:symbolIV}
    \tfrac12 R\big(v,\kappa(s,s)\big)w &=\kappa\big((X_v\beta^{\varphi})(w,s),s\big)+
    \gamma^\varphi\big(\beta^{\varphi}_{v}(s),s\big)w
		\\
		\label{eq:symbolV}
		R(v,w)s &= (X_v\beta^{\varphi})(w,s) - (X_w\beta^{\varphi})(v,s) +
  [\beta^\varphi_v,\beta^\varphi_w](s)
\end{align}
for some $R:\Lambda^2 V\to \fso(V)$ and all $v,w, v_0,\ldots, v_4\in V$, $s\in S'$.
\end{definition}
Combining Theorem \ref{thm:reconstruction}, Corollary \ref{corollary:geometric-symbol} and the discussion 
carried out in this section, we arrive at the following version of the Reconstruction Theorem. Therein 
\begin{equation*}
\label{eq:moduli-abstract-symbols}
\begin{aligned}
\mathcal{AS}&=\frac{\left\{\text{abstract symbols}\;(\varphi,S')\right\}}{\Spin(V)}
\end{aligned}
\end{equation*}
is the moduli space of abstract symbols.
\begin{theorem}[Reconstruction Theorem - Strong Version]
\label{thm:strong-version-reconstruction}
The map $\mathcal SB\longrightarrow\mathcal AS$ that sends a highly supersymmetric supergravity background to its geometric symbol is a $1:1$ correspondence, with image the space of abstract symbols.
\end{theorem}
The construction of supergravity backgrounds 
breaks then into three steps:
\begin{enumerate}
\item Describe $\Spin(V)$-orbits of Lie pairs $(\varphi,S')$ (and therefore the
  associated graded subalgebras $\fa=\fh_{(\varphi,S')}\oplus S'\oplus
  V$ of the Poincar\'e superalgebra);
	\item Check \eqref{eq:symbolII} and \eqref{eq:symbolIII} (roughly speaking, this gives the transvection superalgebra of the background at the 
	level of infinitesimal deformation);
\item See if there exists $R:
  \Lambda^2 V\to\fso(V)$ satisfying \eqref{eq:symbolV} and then check \eqref{eq:symbolIV} (at this stage we fully reconstructed the transvection superalgebra as a deformation).
\end{enumerate}
\begin{remark}
I have some evidence that \eqref{eq:symbolIV} is a consequence of the other identities, but I don't have a complete proof so far.
\end{remark}
\begin{remark}
If $\varphi\in\Lambda^4 V$ appears in a Lie pair $(\varphi,S')$ with $S'=S$ then $\varphi$ is decomposable, cf. \cite{FOFS2017}. It would be desirable to have an a priori understanding of which $\varphi\in\Lambda^4 V$ appear in Lie pairs at all. The variety of such four-vectors is $\SO(V)$-stable and in most likelihood properly contained in $\Lambda^4 V$.
\end{remark}
It is clear that constructing Lie pairs remains the most difficult step, since they are defined by a rather complicated system of coupled algebraic equations, quadratic on $\varphi$ and cubic on $S'$. To a certain extent, this can be regarded as the algebraic counterpart of the bosonic field equations   \eqref{eq:bosfieldeqs} for highly supersymmetric backgrounds.
We here propose a ``separation of variables'' technique to settle it:
\begin{itemize}
	\item For a given $\varphi\in\Lambda^4 V$, we consider the operator $\gamma^\varphi:\odot^2 S\to\fso(V)$ and determine the intersection $$\widetilde\fh=\operatorname{Im}\gamma^\varphi\cap\stab_{\fso(V)}(\varphi)$$ of its {\it full} image with the stabilizer of $\varphi$.
It is easy to see that this is a Lie subalgebra of $\fso(V)$. We also determine the $\widetilde\fh$-submodule $K=(\gamma^{\varphi})^{-1}\big(\stab_{\fso(V)}(\varphi)\big)$ of $\odot^2 S$;
\item We choose a presentation of $S$ adapted to $\widetilde \fh$ and
look for $\widetilde\fh$-submodules $S'$ of $S$. We identify those which satisfy the inclusion $\mathfrak D\subset K$  using $\widetilde\fh$-equivariance.
\end{itemize}
Note that the subspace $S'$ only enters at the last stage and that the inclusion there is a quadratic relation on spinors, and not cubic. A pair $(\varphi,S')$ so obtained is a Lie pair. Indeed
\begin{equation}
\label{eq:hinstab}
\fh=\gamma^{\varphi}(\mathfrak D)\subset\gamma^{\varphi}(K)\subset\stab_{\fso(V)}(\varphi)
\end{equation}
and $\fh\subset\stab_{\fso(V)}(S')$, since $S'$ is an $\widetilde\fh$-module by construction and $\fh\subset\widetilde\fh$ by \eqref{eq:hinstab}. 

In \S\ref{sec:7}, we work out an example where $V=\mathbb R^{1,1}\oplus\mathbb R^9$ splits into the ortohognal direct sum of $\mathbb R^{1,1}=\left\{\be_+,\be_-\right\}$ and $\mathbb R^9=\left\{\be_1,\ldots,\be_9\right\}$, and
$\varphi=\be_+\wedge\phi$ for some $\phi\in\Lambda^3\mathbb R^9$ of small rank. 
The relevant orbits are described in \S\ref{sec:6}.
\section{The $\mathrm{SO}_9(\mathbb R)$-orbits in $\Lambda^3 \mathbb R^9$ of subminimal rank}
\label{sec:6}
In this section, we set $G=\mathrm{SL}_9(\mathbb R)$ and let $G^\theta=\mathrm{SO}_9(\mathbb R)$ be 
the special orthogonal subgroup, i.e., the fixed point set of the Cartan involution $\theta:G\to G$ of $G$. We are interested in the stratification under the action of 
$G^\theta$ of small orbits of $G$ on $\Lambda^3 \mathbb R^9$.

The {\it support} of a trivector $\phi\in\Lambda^3\mathbb R^9$ is the unique minimal subspace $\mathbb E\subset \mathbb R^9$ such that $\phi\in\Lambda^3 \mathbb E$. 
Its dimension is called the {\it rank} of $\phi$ and it is a $G$-invariant, in particular it is one of the simplest $G^\theta$-invariants, together with the trivector's norm.

The trivectors of minimal (non-zero) rank are decomposable and form a $G$-orbit, which is stratified by the level sets of the norm  in a $1$-parameter family of $G^\theta$-orbits. The associated filtered subdeformations have been studied in \cite{FOFS2017}: 
\begin{proposition}
\label{prop:sugra-decomposable}
Let $\phi=\lambda \be_{123}$ for some $\lambda>0$ and $(M,g,F)$ a highly supersymmetric supergravity  background with $symb(M,g,F)=(\varphi=\be_+\wedge \phi,S')$ for some $S'\subset S$. Then $S'=S$  and $(M,g,F)$ is locally isometric to the 
Kowalski-Glikman background. 
\end{proposition}
\begin{proof}
The transvection superalgebra of the Kowalski-Glikman background is the filtered subdeformation with symbol $(\varphi,S)$ \cite{FOFS2017}.
By maximality, the symbol $(\varphi,S')$ of the transvection superalgebra of our background has to coincide with  $(\varphi,S)$.
\end{proof}
The next step in the analysis of filtered deformations and supergravity backgrounds is the rank $5$ orbit, due to the following well-known result:
\begin{lemma}
\label{lem:rank5trivectors}
An indecomposable $\phi\in\Lambda^3\mathbb R^9$ has rank at least $5$. The rank $5$ trivectors constitute a 
 $G$-orbit, with the representative, e.g., $\phi=\be_{123}+\be_{145}$.
\end{lemma}
 This is also a subminimal orbit, in the sense that its Zariski-closure consists of the orbit itself, the minimal orbit of the non-zero decomposable trivectors and the zero trivector, see for instance \cite[page 104]{Martinet1970}. We will therefore denote it by
$\mathcal O^{G}_{submin}$. 
\begin{proposition}[$G^\theta$-orbits of rank $5$ trivectors.]
\label{lem:rank5trivectors-orthogonal}
The subminimal orbit $\mathcal O^{G}_{submin}$ admits a stratification
\begin{equation*}
\label{eq:stratification}
\mathcal O^G_{submin}=\bigcup_{0<\lambda\leq\mu}\mathcal O^{G^\theta}_{\phi_{(\lambda,\mu)}}\;,
\end{equation*}
where $\mathcal O^{G^\theta}_{\phi_{(\lambda,\mu)}}$ is the $G^\theta$-orbit of the trivector $\phi_{(\lambda,\mu)}=\lambda \be_{123}+\mu \be_{145}$. The Lie algebra of the stabiliser $H^\theta$ of $\phi_{(\lambda,\mu)}$ in $G^\theta$ is 
$$
\mathfrak{h}^\theta=
\begin{cases}
\langle \be_{23},\be_{45}\rangle\oplus\mathfrak{so}(\mathbb E^\perp)\cong\mathfrak{u}(1)\oplus\mathfrak{u}(1)\oplus\mathfrak{so}(\mathbb E^\perp)\quad&\text{if}\;\lambda<\mu,\\
\langle \be_{23},\be_{45}, \be_{24}+\be_{35},\be_{25}-\be_{34}\rangle\oplus\mathfrak{so}(\mathbb E^\perp)\cong\mathfrak{u}(2)\oplus\mathfrak{so}(\mathbb E^\perp)\quad&\text{if}\;\lambda=\mu,
\end{cases}
$$
where $\mathbb E=\langle \be_1,\ldots,\be_5\rangle$ is the support of $\phi_{(\lambda,\mu)}$ and $\mathbb E^\perp=\langle \be_6,\ldots,\be_9\rangle$.
\end{proposition}
\begin{proof}
Let $\phi\in\Lambda^3 \mathbb R^9$ be a rank $5$ trivector. 
 Then $g\cdot\phi$ has support $\mathbb E=\langle e_1,\ldots,e_5\rangle$
for some $g\in G^\theta$, and two trivectors with support $\mathbb E$ are in the same $G^\theta$-orbit if and only if they are in the same $\mathrm{O}(\mathbb E)$-orbit. Therefore, it is enough to describe the $\mathrm{O}(\mathbb E)$-orbits of trivectors $\phi\in\Lambda^3 \mathbb E$ of rank $5$. (We will see that they coincide with $\mathrm{SO}(\mathbb E)$-orbits.)

We fix a volume element $\operatorname{vol}\in\Lambda^5\mathbb E^*$ and consider the bijection
$$
\Lambda^3 \mathbb E\to\Lambda^2\mathbb E^*\;,\qquad\phi\mapsto \imath_{\phi}\operatorname{vol}\;,
$$
which is an isomorphism of $\mathrm{SO}(\mathbb E)$-modules. Now $\Lambda^2\mathbb E^*\cong\mathfrak{so}(\mathbb E)$,
so the $\mathrm{SO}(\mathbb E)$-orbits on $\Lambda^3 \mathbb E$ are in bijective correspondence with the adjoint orbits. 
The group $\mathrm{SO}(\mathbb E)$ is compact, hence any 
adjoint orbit has a representative in the Cartan subalgebra 
$$
\mathfrak{t}=\left\{\left(\begin{array}{c|cc|cc}
0 & & & & \\
\hline
& 0 & \mu & &  \\
&-\mu & 0 & &  \\
\hline
& & & 0 & \lambda   \\
& & & -\lambda & 0 
\end{array}
\right)\mid\lambda,\mu\in\mathbb R
\right\}$$ 
of $\mathfrak{so}(\mathbb E)$. Using the Weyl group of $\mathrm{SO}(\mathbb E)$, we may uniquely arrange for $0\leq\lambda\leq\mu$.
The corresponding trivector is $\phi=\lambda \be_{123}+\mu \be_{145}$ and has rank $5$ if and only if $\lambda\neq 0$. 

The last claim follows from a direct computation, which we omit.
\end{proof}
\begin{remark}
A similar analysis for higher rank orbits is possible but it is more involved, as the connection with the adjoint orbits of a compact Lie group is not available in general.
A possible strategy is outlined here:
\begin{enumerate}
	\item Let $\mathcal O^G_\phi\cong G/H$ be the $G$-orbit with representative $\phi\in\Lambda^3\mathbb R^9$ of rank $6\leq k\leq 9$, where 
	$H$ is the stabiliser of $\phi$ in $G$. We then have a stratification
	\begin{equation}
\label{eq:stratificationII}
\mathcal O^G_{\phi}=\bigcup_{i}\mathcal O^{G^\theta}_{\phi_i}
\end{equation}
under the action of $G^\theta$, which is parametrised by the double cosets in $G^\theta \backslash G/H$;
\item The $H$-equivariant map
$G^\theta g\mapsto g^t\cdot g$
identifies $G^\theta \backslash G$ with the space $Sym_9^+$ of positive-definite symmetric matrices with unit determinant.
(Surjectivity follows from Sylvester's law of inertia, the right action of $H$ on $Sym_9^+$ is by congruence.) Hence 
	\begin{equation}
	\label{eq:H-equivariantidentification}
G^\theta \backslash G/H\cong Sym_9^+/H\;
\end{equation}
parametrises the stratification \eqref{eq:stratificationII};
\item Tipically we have a non-trivial Levi decomposition 
\begin{equation}
H=L\ltimes U
\end{equation}
of $H$ with reductive subgroup $L$ and unipotent radical $U$.
To compute $Sym_9^+/H$, we may first consider $Sym_9^+/U$ and then the residual action of $L\cong H/U$ on it.
\end{enumerate}
The orbits $\mathcal O^G_\phi$ of trivectors of rank at most $9$ have been determined in \cite{VE} (upon complexification). 
The description of those of rank $9$ is extremely involved but those of rank $\leq 8$ are automatically nilpotent (in the sense of Vinberg's theory of $\theta$-groups) and relatively few: upon complexification, there are $13$ orbits $\mathcal O^G_\phi$ of rank $8$ and the orbits of rank $\leq 7$ are given by:
\begin{align*}
 \begin{array}{|c|c|c|c|c|c|} \hline
\text{class}  &\text{representative}\;\phi & \text{rank} & \text{reductive part}\;\mathfrak l\;\text{of}\;\fh & \dim\mathfrak l & \dim\fh\\ \hline\hline
 90 &\be_{123}+\be_{147}+\be_{257}+\be_{367}+\be_{456} & 7 & G_2\oplus A_1 & 17 & 31\\ \hline
  93 &\be_{125}+\be_{137}+\be_{247}+\be_{346} & 7 & 3A_1\oplus \CC & 10 & 32\\ \hline
	 94 & \be_{127}+\be_{134}+\be_{256} & 7 & 3A_1\oplus 2\CC & 11 & 35\\ \hline
		95 &	  \be_{125}+\be_{136}+\be_{147}+\be_{234} & 7 & A_2\oplus A_1\oplus \CC & 12 & 38\\ \hline
	96	&		\be_{123}+\be_{456} & 6 & 3A_2 & 24 & 42\\ \hline
				97	&							\be_{124}+\be_{135}+\be_{236} & 6 & 2A_2\oplus\CC & 17 & 43\\ \hline
				99	&	\be_{123}+\be_{145}+\be_{167} & 7 & C_3\oplus A_1\oplus \CC & 25 & 45\\ \hline
				100		&									\be_{123}+\be_{145} & 5 & A_3\oplus C_2\oplus\CC & 26 & 50\\ \hline
					101		&										\be_{123} & 3 & A_5\oplus A_2 & 43 & 61\\ \hline
 \end{array}
 \end{align*}
where a couple of misprints in \cite{VE} have been corrected. Some orbits (but not all) have been investigated in \cite{GH} in the context of pp-waves.

By previous steps, the stratification \eqref{eq:stratificationII} of an orbit $\mathcal O^G_\phi$ is parametrised
by unit volume scalar products up to $H$-equivalence. The study of the associated supergravity backgrounds will be considered in future work,
we now work out our example. We will omit most of the actual details of this calculation.
\end{remark}
\section{Example}
\label{sec:7}
The proof of the following lemma is straightforward.
\begin{lemma}
\label{lemma:invariant-4form}
Let $\varphi\in\Lambda^4 V$ be of the form $\varphi=\be_+\wedge\phi$ for some non-zero $\phi\in\Lambda^3 \mathbb R^9$. Then
$$
\stab_{\fso(V)}(\varphi)=\stab_{\fso(\mathbb R^9)}(\phi)\ltimes (\be_+\wedge \mathbb R^9)\;,
$$
in particular $\stab_{\fso(V)}(\varphi)\subset\stab_{\fso(V)}(\be_+)$.
\end{lemma}
We consider the case where $\phi$ has rank $5$ with maximum stabilizer, i.e., $\phi=\be_{123}+\be_{145}$, up to positive multiples. It is well known that the field equations    \eqref{eq:bosfieldeqs} are invariant under
a homothety that rescales both the metric and the 4-form and that the associated transvection superalgebras
are not isomorphic as filtered subdeformations. However, they are isomorphic if we simply allow for $g\in\CSpin(V)$ in equation \eqref{eq:generaliso}, so we may indeed restrict to $\phi=\be_{123}+\be_{145}$. We recall that $\mathbb R^9=\mathbb E\oplus\mathbb E^\perp$, where $\mathbb E=\langle \be_1,\ldots,\be_5\rangle$ is the support of $\phi$ and $\mathbb E^\perp=\langle \be_6,\ldots,\be_9\rangle$. 

Now $\stab_{\fso(V)}(\varphi)=\big(\mathfrak{u}(2)\oplus\mathfrak{so}(\mathbb E^\perp)\big)\ltimes (\be_+\wedge \mathbb R^9)$
according to Proposition \ref{lem:rank5trivectors-orthogonal} and the image $\operatorname{Im}\gamma^\varphi$ is a $36$-dimensional subspace of $\fso(V)$ satisfying
\begin{align*}
\widetilde\fh&=\operatorname{Im}\gamma^\varphi\cap\stab_{\fso(V)}(\varphi)=\big(\mathfrak{u}(1)\oplus\mathfrak{so}(\mathbb E^\perp)\big)\ltimes (\be_+\wedge \mathbb R^9)\;,
\end{align*}
with $\mathfrak{u}(1)=\mathbb R(\be_{23}+\be_{45})$. The kernel $\Ker\gamma^\varphi$ is $492$-dimensional.

We use the isomorphism of algebras $\Cl(V)\cong \Cl(\mathbb R^{9})\otimes\Cl(\mathbb R^{1,1}) \cong2\mathbb R(16)
\otimes\mathbb R(2)$ (we recall that $\mathbb R^9$ is negative definite for us) to write
$$S\cong\mathbb R^{16}\otimes\mathbb R^2\cong S_+\oplus S_-$$ 
with $S_{\pm}\cong\mathbb R^{16}$ and define the Gamma matrices as block matrices with square blocks of order $16$:
$$
\Gamma_+=\sqrt 2\begin{pmatrix} 0 & \Id\\ 0 & 0\end{pmatrix}\;,\quad \Gamma_-=-\sqrt 2\begin{pmatrix} 0 & 0\\ \Id & 0\end{pmatrix}\;,\quad 
\Gamma_i=\begin{pmatrix} \gamma_i & 0\\ 0 & -\gamma_i\end{pmatrix}\quad(i=1,\ldots,9)\;.
$$
The Gamma matrices $\gamma_i\in\Cl(\mathbb R^{9})$ will be described via a quaternionic formalism.

First of all $\Cl(\mathbb R^5)\cong 2\mathbb H(2)$ and an isomorphism is given by the quaternionic matrices
\begin{align*}
A_1=\begin{pmatrix}
1 & 0 \\
0 &-1
\end{pmatrix}\;,\;
A_2=\begin{pmatrix}
0 & 1 \\
1 & 0
\end{pmatrix}\;,\;
A_3=\begin{pmatrix}
0 & L_i \\
-L_i & 0
\end{pmatrix}\;,\;
A_4=\begin{pmatrix}
0 & L_j \\
-L_j & 0
\end{pmatrix}\;,\;
A_5=\begin{pmatrix}
0 & L_k \\
-L_k & 0
\end{pmatrix}\;,
\end{align*}
where $L_q:\mathbb H\to\mathbb H$ is left multiplication by $q\in\mathbb H$.
We set $\mathbb R^{16}\cong \mathbb H^2\otimes\mathbb R^2$
and define block matrices with square blocks of order $8$ by
\begin{align*}
\gamma_i=\begin{pmatrix} A_i & 0\\ 0 & -A_i\end{pmatrix}
\;,\;
\gamma_6&=\begin{pmatrix} 0 & \Id\\ \Id & 0\end{pmatrix}
\;,\;
\gamma_7=\begin{pmatrix} 0 & R_i\\ -R_i & 0\end{pmatrix}
\;,\;
\gamma_8=\begin{pmatrix} 0 & R_j\\ -R_j & 0\end{pmatrix}
\;,\;
\gamma_9=\begin{pmatrix} 0 & R_k\\ -R_k & 0\end{pmatrix}
\end{align*}
where $i=1,\ldots, 5$. Here $R_q:\mathbb H\to\mathbb H$ is right multiplication and we used the same symbol for its natural action on $\mathbb H^2$.

The subspaces $S_{\pm}\cong \mathbb H^4$ are isotropic w.r.t. the canonical symplectic form, which dually pairs them via the standard inner product on $\mathbb H^4$. Thanks to this and the above Gamma matrices, one easily recovers the Dirac current; here we will just mention 
that $\kappa(S_\pm,S_\pm)=\mathbb R\be_\pm$ and $\kappa(S_+,S_-)=\mathbb R^9$.
\begin{lemma}
The spinorial action of $\widetilde\fh$ is given by
\begin{align}
\be_+\wedge \mathbb R^9&=\left\langle\begin{pmatrix} 0 & \gamma_i \\ 0 & 0\end{pmatrix}\mid i=1,\ldots,9\right\rangle\;,\\
\mathfrak{so}(\mathbb E^\perp)&=\left\langle\begin{pmatrix} \gamma_i\gamma_j & 0 \\ 0& \gamma_i\gamma_j \end{pmatrix}\mid 
6\leq i<j\leq 9\right\rangle\;,\\
\mathfrak{u}(1)&=\left\langle
\begin{pmatrix} \gamma_2\gamma_3+\gamma_4\gamma_5 & 0 \\ 0& \gamma_2\gamma_3+\gamma_4\gamma_5\end{pmatrix}
\right\rangle\;,
\end{align}
where $\gamma_2\gamma_3+\gamma_4\gamma_5=-2{\tiny \left(\begin{array}{c|c}
 \begin{array}{c|c}
  L_i & 0 \\ \hline
  0 & 0
  \end{array}
  & 0 \\
\hline
  0 &
  \begin{array}{c|c}
  L_i & 0 \\ \hline
  0 & 0
  \end{array}
\end{array}\right)}$ as a quaternionic matrix of $\mathbb H^4$. In particular the subspaces of $S$ given by
\begin{align}
S'_1&=S_+\oplus \mathbb H\oplus (0)\oplus\mathbb H\oplus (0)\\
S'_2&=S_+\oplus (0)\oplus \mathbb H\oplus (0)\oplus \mathbb H
\end{align}
are $\widetilde\fh$-stable.
\end{lemma}
Now $K=(\gamma^{\varphi})^{-1}\big(\stab_{\fso(V)}(\varphi)\big)$ clearly contains $\Ker\gamma^\varphi$
and since $\Lambda^1 V\subset\Ker\gamma^\varphi$ we may write 
$
K=\Ker\gamma^\varphi\oplus C
$
for some $C\subset\Lambda^2V\oplus\Lambda^5 V$. Clearly $\dim C=\dim\widetilde\fh=16$ and it turns out that 
\begin{align*}
&\be_{23}+\be_{45}\;,\;\;2\be_{13}+\be_{36789}\;,\;\;-2\be_{12}-\be_{26789}\;,\;\;2\be_{15}+\be_{56789}\;,\;\;-2\be_{14}-\be_{46789}\;,\\
&\be_{45789}+\be_{23789}\;,\;\;\be_{45689}+\be_{23689}\;,\;\;\be_{45679}+\be_{23679}\;,\;\;\be_{45678}+\be_{23678}\;\;\;\text{and}\\
&\be_-\wedge\big\{-2\be_{1}+\be_{6789},\;\be_{4589}+\be_{2389},\;\be_{4579}+\be_{2379},\;\be_{4578}+\be_{2378},\;\be_{4569}+\be_{2369}\\
&\be_{4568}+\be_{2368},\;\be_{4567}+\be_{2367}\big\}
\end{align*}
are generators of $C$. Using 
this decomposition one may check that $\odot^2 S'\subset K$ in both cases -- in particular $\mathfrak D\subset K$ -- and that $S'$ cannot be enlarged
(otherwise one constructs $s\odot t\in\odot^2 S'$ with vanishing Dirac current and such that $\gamma^\varphi(s,t)\notin\widetilde\fh$).
\begin{corollary}
The pairs $(\varphi,S'_1)$ and $(\varphi,S'_2)$ are Lie pairs, and they are maximal.
\end{corollary}
The Dirac kernel $\mathfrak D$ has dimension $\dim\mathfrak D=\dim \odot^2 S'-\dim V=289$ in both cases
but the two spaces are different; indeed we have
\begin{align}
\fh=\gamma^{\varphi}(\mathfrak D)=
\begin{cases}
\widetilde\fh=\big(\mathfrak{u}(1)\oplus\mathfrak{so}(\mathbb E^\perp)\big)\ltimes (\be_+\wedge \mathbb R^9)\;\;\text{if}\;\;S'=S'_1,\\
\mathfrak{u}(1)\ltimes (\be_+\wedge \mathbb E^\top)\;\;\text{if}\;\;S'=S'_2,
\end{cases}
\end{align}
where $\mathbb E^\top=\langle\be_2,\be_3,\be_4,\be_5\rangle\subset \mathbb E$. In the first case
$$
X(v)=\gamma^{\varphi}(\Sigma v)\subset \gamma^{\varphi}(\odot^2 S')\subset\widetilde \fh=\fh
$$
for any section $\Sigma:V\to\odot^2 S'$, so we may arrange $X=0$ without loss of generality. It turns out that  $\gamma^{\varphi}(\odot^2 S')\subset\fh$ in the second case too, so $X=0$. In particular $d\varphi=0$.

Using Gamma matrices, it is straightforward to see 
$$\beta_{\be_+}=0\;,\qquad\beta_{\be_i}=\begin{pmatrix} 0 & *\\ 0 & 0\end{pmatrix}\quad(i=1,\ldots,9)\;,$$ 
as endomorphisms of $S=S_+\oplus S_-$, and that $\beta_{\be_-}$ preserves the decomposition of $S$ into eight copies of $\mathbb H$.
It follows that $\beta_v(S')\subset S'$ for all $v\in V$ in both cases.

We are left with the identities involving the curvature tensor $R:\Lambda^2 V\to\fso(V)$. To solve them, we
decompose 
$$V=\mathbb R\be_+\oplus\mathbb R\be_-\oplus\mathbb R\be_1\oplus\mathbb E^\top\oplus\mathbb E^\perp$$
and accordingly write $v=v_++v_-+v_1+v_{\top}+v_\perp$ 
for any $v\in V$. The general strategy to solve \eqref{eq:symbolV}
is analogous to \cite[\S 4.3]{FOFS2017}, although a tad more involved: we give the final results and briefly comment
on how we derived them. In the first and the second case, respectively, we get:
\begin{equation}
\label{eq:curvature}
\begin{aligned}
	R(v,w)&=-\eta(\be_+,v_-)\big(\tfrac 49 w_1\wedge \be_+ +\tfrac{1}{36}w_\top\wedge\be_+ +\tfrac{1}{9}w_\perp\wedge\be_+\big)\\
	&\;\;\;+\eta(\be_+,w_-)\big(\tfrac 49 v_1\wedge \be_+ +\tfrac{1}{36}v_\top\wedge\be_+ +\tfrac{1}{9}v_\perp\wedge\be_+\big)\;,\\
	R(v,w)&=-\tfrac14\eta(\be_+,v_-)w_\top\wedge\be_+ +\tfrac14\eta(\be_+,w_-)v_\top\wedge\be_+\;,
	\end{aligned}
	\end{equation}
	for all $v,w\in V$. A subtle point to obtain equations \eqref{eq:curvature} is that
	the R.H.S. of \eqref{eq:symbolV} makes sense for all $s\in S$ but the action in there is {\it not} that of an element of $\fso(V)$. Upon restriction to $S'$, one may re-absorb to $\Lambda^2 V$ the contributions coming from the elements in
	$\End(S)\cong\bigoplus_{0\leq p\leq 5}\Lambda^p V$
	with $p\neq 2$, and this is how we arrived at \eqref{eq:curvature}. Checking equation \eqref{eq:symbolIV} is then a direct matter, if tedious, of Clifford identities, as in e.g. \cite[Theorem 16]{FOFS2017}.

We recovered two symmetric pp-waves backgrounds that were discovered in \cite{GH}, one
indecomposable and the other decomposable. The proof of the existence comes at the same time with
the construction of their transvection superalgebras -- they are the filtered deformations
of $\fa=\fh\oplus S'\oplus V$ determined by $\varphi$ -- in particular it comes with the knowledge that 
$N=24$ is the amount of Killing spinors  preserved.
\begin{acknowledgement}
A.S. is grateful to G. Carnovale, F. Esposito and J. Figueroa-O'Farrill for useful discussions.
The research of A. S. is supported by 
the project ``Supergravity backgrounds and Invariant theory'' at the University of Padova and partly supported
by project BIRD179758/17 ``Stratifications in algebraic groups, spherical varieties, Kac Moody algebras
and Kac Moody groups'' and project DOR1717189/17 ``Algebraic, geometric and combinatorial properties of conjugacy classes''.
\end{acknowledgement}
%

\end{document}